\documentclass[envcountsame,envcountsect]{llncs}

\usepackage{amsmath,amssymb}
\usepackage{tikz}
\usetikzlibrary{positioning,arrows,automata}
\usepackage{wrapfig}
\usepackage{stmaryrd}   
\usepackage[all]{xy}
\usepackage{xspace}
\usepackage{mathtools}

\newcommand{\may}[1]{\stackrel{#1}{\dashrightarrow}}	
\newcommand{\must}[1]{\stackrel{#1}{\longrightarrow}}	

\newcommand{\omay}{\mathord{\may{}}}
\newcommand{\omust}{\mathord{\must{}}}

\newcommand{\Act}{\Sigma}

\newcommand{\true}{\mathbf{tt}}
\newcommand{\false}{\mathbf{ff}}
\newcommand{\B}{\mathcal B}

\newcommand{\xmts}{\mathcal{M}}
\newcommand{\deth}{\mathcal{D}}
\newcommand{\parh}{\mathcal{P}}

\newcommand{\mr}{\ensuremath{\le_\mathrm{m}}}
\newcommand{\tr}{\ensuremath{\le_\mathrm{t}}}
\newcommand{\semantics}[1]{\llbracket#1\rrbracket}

\newcommand{\theoremlike}[2]{\par\medskip\penalty-250\refstepcounter{theorem}{{\bfseries\noindent#2
\ref{#1}.}}}
\newcommand{\thmhelperpre}[2]{\theoremlike{#1}{#2}}
\newcommand{\thmhelperpost}{\par\medskip}

\newcommand{\tuple}[1]{\ensuremath{ ( #1  )}}
\newcommand{\app}[2]{\ensuremath{{#1}\left({#2}\right)}}

\newcommand{\Tran}{\mathrm{Tran}}

\newcommand{\red}{\ensuremath{\mathit{red}}\xspace}
\newcommand{\green}{\ensuremath{\mathit{green}}\xspace}
\newcommand{\yellow}{\ensuremath{\mathit{yellow}}\xspace}
\newcommand{\yellowRed}{\ensuremath{\mathit{yellowRed}}\xspace}
\newcommand{\go}{\ensuremath{\mathit{go}}\xspace}
\renewcommand{\stop}{\ensuremath{\mathit{stop}}\xspace}
\newcommand{\getReady}{\ensuremath{\mathit{ready}}\xspace}
\newcommand{\prepareToStop}{\ensuremath{\mathit{ready}}\xspace}
\newcommand{\mustGoYellow}{\ensuremath{\mathit{reqYfromG}}\xspace}
\newcommand{\mustGoYellowRed}{\ensuremath{\mathit{reqYfromR}}\xspace}
\newcommand{\mustGoYellowBoth}{\ensuremath{\mathit{reqYellow}}\xspace}

\newcommand{\trafficlight}[5]{
  \begin{scope}[xshift=#4cm,yshift=#5cm,scale=0.74,transform shape]
    \draw [rounded corners=2pt,fill=gray!80,thick] (-.28,-.70) 
    rectangle (.28,.70);
    \draw [fill=#1] (0,.44) circle (6pt);
    \draw [fill=#2] (0,0) circle (6pt);
    \draw [fill=#3] (0,-.44) circle (6pt);
  \end{scope}}

\colorlet{yellow}{yellow!81!black}

\newcommand{\todo}[1]{%
  \begin{tikzpicture}[remember picture]%
      \node [coordinate] (inText) {};%
  \end{tikzpicture}%
  \marginpar{%
      \begin{tikzpicture}[remember picture]%
          \draw node[draw=red, color=red, text width = 3.2cm] (inNote)%
                   {{\small #1}};%
      \end{tikzpicture}%
  }%
  \begin{tikzpicture}[remember picture, overlay]%
      \draw[draw = red, thick]%
          ([yshift=-0.1cm] inText)%
              -| ([xshift=-0.1cm] inNote.west)%
              -| (inNote.west);%
  \end{tikzpicture}%
}%

\renewcommand{\todo}[1]{}
\title{On Refinements of Boolean and Parametric\\ Modal Transition Systems }

\author{Jan K\v ret\'insk\'y\inst{1,2} \and Salomon Sickert\inst{1}}

\institute{
Institut f\"ur Informatik, Technische Universit\"at M\"unchen, Germany\\
\and
Faculty of Informatics, Masaryk University, Brno, Czech Republic
}

\begin{document}

\pagestyle{plain}

\maketitle

\begin{abstract}
We consider the extensions of modal transition systems (MTS), na\-mely Boolean MTS and parametric MTS and we investigate the refinement problems over both classes. Firstly, we reduce the problem of modal refinement over both classes to a problem solvable by a QBF solver and provide experimental results showing our technique scales well. Secondly, we extend the algorithm for thorough refinement of MTS providing better complexity then via reductions to previously studied problems. Finally, we investigate the relationship between modal and thorough refinement on the two classes and show how the thorough refinement can be approximated by the modal refinement.
\end{abstract}

\section{Introduction}\label{sec:intro}

Due to the ever increasing complexity of software systems and their reuse, component-based design and verification have become crucial. Therefore, having a~specification   formalism   that   supports \emph{component-based}  development and \emph{stepwise  refinement} is very useful. In such a framework, one can start from  an initial specification, proceed with a~series of small  and   successive  refinements  until  eventually
a~specification is reached from which an implementation can be extracted directly.  In each  refinement  step, we can replace a~single component  of the current specification with a~more concrete/implementable one.  The correctness of such a step  should follow from the correctness  of the  refinement  of the  replaced  component, so that the methodology supports \emph{compositional} verification.

\emph{Modal transition  systems} (MTS)  were introduced  by   Larsen  and Thomsen~\cite{DBLP:conf/lics/LarsenT88} in order to obtain an operational, yet expressive and manageable  specification formalism meeting the  above properties.
Their success resides in natural combination of two features. Firstly, it is the simplicity of labelled transition systems, which have proved appropriate for behavioural description of systems as well as their compositions; MTS as their extension inherit this appropriateness. Secondly, 
as opposed to e.g.~temporal logic specifications, MTS can be easily \emph{gradually refined} into implementations while preserving the desired behavioural properties. In this work, we focus on checking the refinement between MTS and also their recent extensions.

The formalism of MTS has proven to be useful in practice. Industrial applications are as old as~\cite{DBLP:journals/scp/Bruns97} where MTS have been used for an air-traffic system at Heathrow airport. Besides, MTS are advocated as an appropriate base for interface theories in~\cite{RB-acsd09} and for product line theories in~\cite{nyman2008modal}. Further, MTS based software engineering methodology for design via merging partial descriptions of behaviour has been established in~\cite{DBLP:conf/sigsoft/UchitelC04}. Moreover, the tool support is quite extensive, e.g.~\cite{DBLP:journals/fmsd/BorjessonLS95,DBLP:conf/eclipse/DIppolitoFFU07,DBLP:conf/atva/BauerML11,DBLP:conf/atva/BenesCK11}.

MTS consist of a set of states and two transition relations. The \emph{must} transitions prescribe which behaviour has to be present in every refinement of the system; the \emph{may} transitions describe the behaviour that is allowed, but need not be realized in the refinements. This allows for underspecification of non-critical behaviour in the early stage of design, focusing on the main properties, verifying them and sorting out the details of the yet unimplemented non-critical behaviour later.

Over the years, many extensions of MTS have been proposed. While MTS can only specify whether or not a particular transition is required, some extensions equip MTS with more general abilities to describe what \emph{combinations} of transitions are possible. Disjunctive MTS (DMTS)~\cite{DBLP:conf/lics/LarsenX90}\todo{\cite{DBLP:conf/atva/BenesCK11}} can specify that at least one of a given set of transitions is present. One selecting MTS~\cite{DBLP:journals/jlp/FecherS08} allow to choose exactly one of them. Boolean MTS (BMTS)~\cite{DBLP:conf/atva/BenesKLMS11}\todo{cite OTS} cover all Boolean combinations of transitions. The same holds for acceptance automata~\cite{TheseJBR07} and Boolean formulae with states~\cite{bfs}, which both express the requirement by listing all possible sets instead of a Boolean formula. Parametric MTS (PMTS)~\cite{DBLP:conf/atva/BenesKLMS11} add parameters on top of it, so that we can also express persistent choices of transitions and relate possible choices in different parts of a system. This way, one can model hardware dependencies of transitions and systems with prices~\cite{DBLP:conf/lpar/BenesKLMS12}.

\paragraph{\bf Our contribution} 
In this paper, we investigate extensions of MTS with respect to two notions of refinement. The \emph{modal refinement} is a syntactically defined notion extending on the one hand bisimulation and on the other hand simulation. Similarly to bisimulation having a counterpart in trace equivalence, here the counterpart of modal refinement is the \emph{thorough refinement}. It is the corresponding semantically defined notion relating (by inclusion) the sets of implementations of the specifications. 

We focus both on theoretical and practical complexity of the refinement problems. While modal refinement on MTS and disjunctive MTS can be decided in polynomial time, on BMTS and PMTS it is higher in the polynomial hierarchy ($\Pi_2$ and $\Pi_4$, respectively). The huge success of SAT and also QBF solvers inspired us to reduce these refinement problems to problems solvable by a QBF solver. We have also performed experimental results showing that this solution scales well in the size of the system as well as in the number of parameters, while a direct naive solution is infeasible.

Further, we extend the decision algorithm for thorough refinement checking over MTS~\cite{DBLP:journals/iandc/BenesKLS12} and DMTS~\cite{DBLP:conf/atva/BenesCK11-techrep} to the setting of BMTS and PMTS. We show how PMTS can be translated to BMTS and BMTS can then be transformed to DMTS. As we can decide the problem on DMTS in EXPTIME, this shows decidability for BMTS and PMTS, but each of the translations is inevitably exponential. However, we show better upper bounds than doubly and triply exponential. To this end, we give also a direct algorithm for showing the problem is in NEXPTIME for BMTS and 2-EXPTIME for PMTS.


Since the thorough refinement is EXPTIME-hard for already MTS, it is harder than the modal refinement, which is in P for DMTS and in $\Pi_4$ for PMTS. Therefore, we also investigate how the thorough refinement can be approximated by the modal refinement. While underapproximation is easy, as modal refinement implies thorough refinement, overapproximation is more difficult. Here we extend our method of the deterministic hull for MTS~\cite{DBLP:journals/tcs/BenesKLS09} to both BMTS and PMTS. We prove that for BMTS modal and thorough refinements coincide if the refined system is deterministic, which then yields an overapproximation via the deterministic hull. Finally, in the case with PMTS, we need to overapproximate the behaviour dependent on the parameters, because the coincidence of the refinements on deterministic systems fails for PMTS.

Our contribution can be summarized as follows:
\begin{itemize}
 \item We reduce the problem of modal refinement over BMTS and PMTS to a problem solvable by a QBF solver. We provide promising experimental results showing this solution scales well.
 \item We extend the algorithm for thorough refinement on MTS and DMTS to BMTS and PMTS providing better complexity then via translation of these formalisms to DMTS. This also shows (together with results on modal refinement) that we can make use of the more compact representation used in the formalisms of BMTS and PMTS. 
 \item We investigate the relationship between modal and thorough refinement on BMTS and PMTS. We introduce approximation methods for the thorough refinement on BMTS and PMTS through the modal refinement.
\end{itemize}

\paragraph{\bf Related work} 
There are various other approaches to deal with
component \emph{refinements}. They range from
subtyping~\cite{DBLP:journals/toplas/LiskovW94} over Java modelling
language~\cite{DBLP:conf/fase/JacobsP01} to interface theories close to MTS
such as interface automata~\cite{DBLP:conf/sigsoft/AlfaroH01}. Similarly to
MTS, interface automata are behavioural interfaces for components. However, their
composition works very differently. Furthermore, its notion of refinement is
based on alternating simulation~\cite{DBLP:conf/concur/AlurHKV98}, which has
been proved strictly less expressive than MTS refinement---actually coinciding
on a subclass of MTS---in the paper~\cite{DBLP:conf/esop/LarsenNW07}, which
combines MTS and interface automata based on I/O automata~\cite{lynch:io88}.
The compositionality of this combination is further investigated
in~\cite{RBBCLP10}.

Further, opposite to the design of correct software where an abstract verified
MTS is transformed into a concrete implementation, one can consider checking
correctness of software through \emph{abstracting} a concrete implementation into 
a~coarser system. The use of MTS as abstractions has been advocated
e.g.~in~\cite{DBLP:conf/concur/GodefroidHJ01}. While usually overapproximations
(or underapproximations)  of systems are constructed and thus only purely
universal (or existential) properties can be checked,
\cite{DBLP:conf/concur/GodefroidHJ01} shows that using MTS one can check mixed
formulae (arbitrarily combining universal and existential properties) and,
moreover, at the same cost as checking universal properties using traditional
conservative abstractions. This advantage has been investigated also in the
context of systems equivalent or closely related to
MTS~\cite{DBLP:conf/esop/HuthJS01,DBLP:journals/toplas/DamsGG97,DBLP:conf/cav/Namjoshi03,DBLP:conf/lics/DamsN04,DBLP:conf/atva/CampetelliGLT09,DBLP:conf/popl/GodefroidNRT10}.


MTS can also be viewed as a fragment of mu-calculus that is ``graphically
representable''~\cite{DBLP:conf/caap/BoudolL90,bfs}. The graphical representability
of a variant of alternating simulation called covariant-contravariant
simulation has been recently studied
in~\cite{DBLP:journals/corr/abs-1108-4464}.

\paragraph{\bf Outline of the paper} In Section~\ref{sec:def-mts}, we recall the formalism of MTS and the extensions discussed. Further, in Section~\ref{sec:def-ref}, we recall the modal refinement problem. We reduce it to a QBF problem in Section~\ref{sec:mr}. In Section~\ref{sec:tr}, we give a solution to the thorough refinement problems. Section~\ref{sec:mrtr} investigates the relationship of the two refinements and how modal refinement can approximate the thorough refinement. 
We conclude in Section~\ref{sec:concl}.




\section{Modal Transition Systems and Boolean and Parametric Extensions}\label{sec:def-mts}

In this section, we introduce the studied formalisms of modal transition systems and their Boolean and parametric extensions. We first recall the standard definition of MTS:

\begin{definition}
A \emph{modal transition system (MTS)} over an action alphabet $\Sigma$ is a~triple $(S,\omay,\omust)$, where $S$ is a
set of \emph{states} and $\omust\subseteq\omay\subseteq S\times
\Act\times S$ are \emph{must} and \emph{may} transition relations,
respectively.
\end{definition}

The MTS are often drawn as follows. Unbroken arrows denote the must (and underlying may) transitions while dashed arrows denote may transitions where there is no must transition.

\begin{example}
The MTS on the right is adapted from~\cite{DBLP:conf/atva/BenesKLMS11} and models traffic lights of types used e.g.~in Europe and in North America. In state \green  
\end{example}
\begin{wrapfigure}{r}{0.3\textwidth}
\begin{center}
\vspace*{-3em}
\begin{tikzpicture}[->,>=stealth',initial text=,xscale=0.705,yscale=0.70,transform shape]
    \tikzstyle{every node}=[font=\large] \tikzstyle{every
      state}=[fill=black,shape=rectangle,inner sep=.5mm,minimum height=12mm, minimum width=6mm]
    
    \node[
state,draw=white,fill=white] (green2) at (1,0) {};
    \trafficlight{white}{white}{green}{1}{0}
    
    \node[state,draw=white,fill=white] (red2) at (5,0) {};
    \trafficlight{red}{white}{white}{5}{0}
    
    \node[state,draw=white,fill=white] (yellow2) at (3,-2) {};
    \trafficlight{white}{yellow}{white}{3}{-2}
    
    \node[state,draw=white,fill=white] (yellowred2) at (3,2) {};
    \trafficlight{red}{yellow}{white}{3}{2}
    
    \path (red2) edge[dashed] node [above] {\go} (green2);
    \path (red2.north) edge[dashed,bend right] node [above,sloped] {\getReady} (yellowred2.east);
    \path (yellowred2.west) edge [bend right] node [above,sloped] {\go} (green2.north);
    \path (green2.south) edge [bend right] node [below,sloped] {\prepareToStop} (yellow2.west);
    \path (yellow2.east) edge [bend right] node [below,sloped] {\stop} (red2.south);
\end{tikzpicture}
\vspace*{-4em}
\end{center}
\end{wrapfigure}

\vspace*{-0.95em}
\noindent
on the left there is a must transition under \getReady to state \yellow from which there is must transition to \red. Here transitions to \yellowRed and back to \green are may transition. Intuitively, this means that any final implementation may have one or the other transition or both or none. In contrast, the must transitions are present in all implementations.
 
%
%
%
%
%

Note that using MTS, we cannot express the set of implementations with exactly one of the transitions in \red. For that, we can use Boolean MTS~\cite{DBLP:conf/atva/BenesKLMS11} instead, which can express not only arbitrary conjunctions and disjunctions, but also negations and thus also exclusive-or. However, in Boolean MTS it may still happen that at first only transition to \green is present, but in the next round of the traffic lights cycle only the transition to \yellowRed is present. To make sure the choice will remain the same in the whole implementation, parametric MTS have been introduced~\cite{DBLP:conf/atva/BenesKLMS11} extending the Boolean MTS.

Before we define the most general class of parametric MTS and derive other classes as special cases, we first recall the standard propositional logic.
A~Boolean formula over a set $X$ of atomic propositions is given by
the following abstract syntax 
\[ 
\varphi ::= \true \mid x \mid \neg \varphi \mid \varphi \wedge \psi 
\mid \varphi \vee \psi \]
where $x$ ranges over $X$. 
The set of all Boolean formulae over the set $X$ is denoted by $\B(X)$. 
Let $\nu \subseteq X$ be a valuation, i.e.~a~set of variables with value true, then the satisfaction relation $\nu \models \varphi$ is given by
$\nu \models \true$, $\nu \models x$ iff $x \in \nu $,
and the satisfaction of the remaining Boolean connectives is defined in the standard way.
We also use the standard derived operators like
exclusive-or  
$\varphi \oplus \psi := (\varphi \wedge \neg \psi) \vee (\neg \varphi 
\wedge \psi)$, implication
$\varphi \Rightarrow \psi := \neg \varphi \vee \psi$ and equivalence $\varphi \Leftrightarrow \psi := (\neg \varphi \vee \psi)\wedge(\varphi \vee \neg \psi)$.

We can now proceed with the definition of parametric MTS. In essence, it is a labelled transition system where we can specify which transitions can be present depending on values of some fixed parameters.

\begin{definition}
A~\emph{parametric modal transition system (PMTS)} over an action alphabet $\Sigma$ is 
a~tuple $(S,T,P,\Phi)$ where
\begin{itemize}
\item 
$S$ is a~set of \emph{states},
\item 
$T \subseteq S \times \Sigma \times S$ is 
a~\emph{transition relation}, 
\item 
$P$ is a~finite set of \emph{parameters}, and
\item 
$\Phi : S \to \mathcal B((\Sigma \times S) \cup P)$ is an 
\emph{obligation function} over the outgoing transitions and parameters. We assume that whenever $(a,t)$ occurs in $\Phi(s)$ then 
$(s,a,t) \in T$.
\end{itemize}
A \emph{Boolean modal transition system (BMTS)} is a PMTS with the set of parameters $P$ being empty. A \emph{disjunctive MTS (DMTS)} is a BMTS with the obligation function in conjunctive normal form and using no negation.
%
%
An \emph{implementation} (or \emph{labelled transition system}) is a BMTS with
$\Phi(s) = \bigwedge_{(s,a,t) \in T}(a,t)$ for each $s \in S$. 
\end{definition}

An MTS is then a BMTS with $\Phi(s)$ being a~conjunction of positive literals (some of the outgoing transitions), for each $s \in S$. More precisely, $\omay$ is the same as $T$, and $(s,a,t)\in\omust$ if and only if $(a,t)$ is one of the conjuncts of $\Phi(s)$.

\begin{example}
An example of a PMTS which captures the traffic lights used e.g.~in Europe for cars and for pedestrians is depicted below. Depending on the valuation of parameter \mustGoYellowBoth, we either always use the yellow light between the red and green lights, or we never do. The transition relation is depicted using unbroken arrows.

\smallskip

\begin{tikzpicture}[->,>=stealth',initial text=,xscale=0.705,yscale=0.70,transform shape]
    \tikzstyle{every node}=[font=\large] \tikzstyle{every
      state}=[fill=black,shape=rectangle,inner sep=.5mm,minimum height=12mm, minimum width=6mm]

    \node[
state,draw=white,fill=white] (green2) at (8,0) {};
    \trafficlight{white}{white}{green}{8}{0}
    
    \node[state,draw=white,fill=white] (red2) at (12,0) {};
    \trafficlight{red}{white}{white}{12}{0}
    
    \node[state,draw=white,fill=white] (yellow2) at (10,-2) {};
    \trafficlight{white}{yellow}{white}{10}{-2}
    
    \node[state,draw=white,fill=white] (yellowred2) at (10,2) {};
    \trafficlight{red}{yellow}{white}{10}{2}
    
    \path (red2) edge [bend right=10] node [above] {\go} (green2);
    \path (green2) edge [bend right=10] node [below] {\stop} (red2);
    \path (red2.north) edge [bend right] node [above,sloped] {\getReady} (yellowred2.east);
    \path (yellowred2.west) edge [bend right] node [above,sloped] {\go} (green2.north);
    \path (green2.south) edge [bend right] node [below,sloped] {\prepareToStop} (yellow2.west);
    \path (yellow2.east) edge [bend right] node [below,sloped] {\stop} (red2.south);
    
    \node [right] at (7+7.0,2) {Parameters: $P=\{\mustGoYellowBoth\}$};
    \node [right] at (7+7,1) {Obligation function:};
    
    \node [right] at (7+7.7,.5) {$\Phi(\green) = ((\stop,\red) \oplus (\prepareToStop,\yellow))$};
    \node [right] at (7+10.1,0) {$\wedge ( \mustGoYellowBoth \Leftrightarrow (\prepareToStop,\yellow)) $};
    \node [right] at (7+7.7,-.5) {$\Phi(\yellow) =  (\stop,\red)$};
    \node [right] at (7+7.7,-1) {$ \Phi(\red) = ((\go,\green) \oplus (\getReady,\yellowRed))$};
    \node [right] at (7+10.1,-1.5) {$ \wedge ( \mustGoYellowBoth \Leftrightarrow (\getReady,\yellowRed))$};
    \node [right] at (7+7.7,-2) {$\Phi(\yellowRed) =  (\go,\green)$};
\end{tikzpicture}  
\end{example}




\section{Modal Refinement}\label{sec:def-ref}


A fundamental advantage of MTS-based formalisms is the presence of 
\emph{modal refinement} that allows for a step-wise system design (see e.g.~\cite{AHLNW:EATCS:08}).
We start with the standard definition of modal refinement for MTS and then discuss extensions to BMTS and PMTS.

\begin{definition}[MTS Modal Refinement]\label{def:mr-mts}
For states $s_0$ and $t_0$ of MTS $(S_1,\omust_1,\omay_1)$ and $(S_2,\omust_2,\omay_2)$, respectively, we say that $s_0$ \emph{modally refines} $t_0$, written $s_0\mr t_0$, if $(s_0,t_0)$ is contained in a~relation $R\subseteq S_1\times S_2$ satisfying for every $(s,t)\in R$ and every $a \in \Act$:
\begin{enumerate}
  \item if $s\may{a}_{1}s'$
    then there is a~transition $t\may{a}_{2}t'$
    with~$(s',t')\in R$, and
  \item if $t\must{a}_{2}t'$
    then there is a~transition $s\must{a}_{1}s'$
    with~$(s',t')\in R$.
\end{enumerate}
\end{definition}
Intuitively, $s\mr t$ iff whatever $s$ can do is allowed by $t$ and whatever $t$ requires can be done by $s$. Thus $s$ is a refinement of $t$, or $t$ is an abstraction of $s$.  Further, an \emph{implementation of} $s$ is a state of an implementation (labelled transition system) with $i\mr s$.

In~\cite{DBLP:conf/atva/BenesKLMS11}, the modal refinement has been extended to PMTS (and thus BMTS) so that it coincides on MTS. We first recall the definition for BMTS. To this end, we set the following notation.
Let $(S,T,P,\Phi)$ be a~PMTS and $\nu \subseteq P$ be a valuation. For $s \in S$, we write 
$T(s)=\{(a,t)\mid (s,a,t)\in T\}$ and denote by $$\Tran_{\nu}(s) = \{ E\subseteq T(s) \mid E \cup \nu \models \Phi(s)\}$$ the set of all admissible sets of transitions from $s$ under the fixed truth values of the parameters. In the case of BMTS, we often write $\Tran$ instead of $\Tran_\emptyset$.

\begin{definition}[BMTS Modal Refinement]\label{def:mr-bmts}
For states $s_0$ and $t_0$ of BMTS $(S_1,T_1,\emptyset,\Phi_1)$ and $(S_2,T_2,\emptyset,\Phi_2)$, respectively, we say that $s_0$ \emph{modally refines} $t_0$, written $s_0\mr t_0$, if $(s_0,t_0)$ is contained in a relation $R\subseteq S_1\times S_2$ satisfying
for every $(s,t) \in R$: 
\begin{align*}  
 \forall M \in \Tran(s) : \exists N \in \Tran(t) : 
        & \ \ \forall (a,s') \in M : \exists (a,t') \in N : (s',t') \in R \ \ \wedge \\
                &\ \ \forall (a,t') \in N : \exists (a,s') \in M : (s',t') \in R\ .
\end{align*}
\end{definition}

For PMTS, we propose here a slightly altered definition, which corresponds more to the intuition, is closer to the semantically defined notion of thorough refinement, but still keeps the same complexity as established in~\cite{DBLP:conf/atva/BenesKLMS11}. We use the following notation.
For a PMTS $\xmts=(S,T,P,\Phi)$, a valuation $\nu\subseteq P$ of parameters induces a BMTS $\xmts^\nu=(S,T,\emptyset,\Phi')$ where each occurrence of $p\in\nu$ in $\Phi$ is replaced by $\true$ and of $p\notin\nu$ by $\neg\true$, i.e. $\Phi'(s)=\Phi(s)[\true/p \text{ for } p\in\nu,\false/p \text{ for } p\notin\nu]$ for each $s\in S$.
We extend the notation to states and let $s^\nu$ denote the state of $\xmts^\nu$ corresponding to the state $s$ of $\xmts$.

\begin{definition}[PMTS Modal Refinement]\label{def:mr-pmts}
For states $s_0$ and $t_0$ of PMTS $(S_1,T_1,P_1,\Phi_1)$ and $(S_2,T_2,P_2,\Phi_2)$, we say that $s_0$ \emph{modally refines} $t_0$, written $s_0\mr t_0$,
if for every $\mu \subseteq P_1$ there exists $\nu \subseteq P_2$ such that $s_0^\mu\mr t_0^\nu$.
\end{definition}

%

Before we comment on the difference to the original definition, we illustrate the refinement on an example of~\cite{DBLP:conf/atva/BenesKLMS11} where both definitions coincide. 
\begin{example}
Consider the rightmost PMTS below.
It has two parameters, namely
\mustGoYellow and \mustGoYellowRed whose values can be set independently
and it can be refined by the system in the middle of the figure having
only one parameter \mustGoYellowBoth. This single parameter simply binds 
the two original parameters to the same value. The PMTS in the middle
can be further refined into the implementations where either 
yellow is always used in both cases, or never at all as discussed in the previous example.
Up to bisimilarity, the \green state of this system only has the two implementations on the left.

\begin{center}
\scalebox{0.8}{
\begin{tikzpicture}[->,>=stealth',initial text=,xscale=0.705,yscale=0.70,transform shape]
    \tikzstyle{every node}=[font=\large] \tikzstyle{every
      state}=[fill=black,shape=rectangle,inner sep=.5mm,minimum height=12mm, minimum width=6mm]
    
    \node[initial,state,draw=white,fill=white] (green3) at (8,0) {};
    \trafficlight{white}{white}{green}{8}{0}
    
    \node[state,draw=white,fill=white] (red3) at (12,0) {};
    \trafficlight{red}{white}{white}{12}{0}
    
    \node[state,draw=white,fill=white] (yellow3) at (10,-2) {};
    \trafficlight{white}{yellow}{white}{10}{-2}
    
    \node[state,draw=white,fill=white] (yellowred3) at (10,2) {};
    \trafficlight{red}{yellow}{white}{10}{2}
    
    \path (red3) edge [bend right=10] node [above] {\go} (green3);
    \path (green3) edge [bend right=10] node [below] {\stop} (red3);
    \path (red3.north) edge [bend right] node [above,sloped] {\getReady} (yellowred3.east);
    \path (yellowred3.west) edge [bend right] node [above,sloped] {\go} (green3.north);
    \path (green3.south) edge [bend right] node [below,sloped] {\prepareToStop} (yellow3.west);
    \path (yellow3.east) edge [bend right] node [below,sloped] {\stop} (red3.south);
    
    \node [right] at (6.6,3.2) {Parameters: $P=\{\mustGoYellowRed,\mustGoYellow\}$};
    \node [right] at (6.6,-3.5) {Obligation function:};
    \node [right] at (6.6,-6.5) {$\Phi(\yellowRed) =  (\go,\green)$};
    \node [right] at (6.6,-4) {$\Phi(\green) = ((\stop,\red) \oplus (\prepareToStop,\yellow))$};
    \node [right] at (7.2,-4.5) {$\wedge ( \mustGoYellow \Leftrightarrow (\prepareToStop,\yellow)) $};
    \node [right] at (6.6,-5) {$\Phi(\yellow) =  (\stop,\red)$};
    \node [right] at (6.6,-5.5) {$ \Phi(\red) = ((\go,\green) \oplus (\getReady,\yellowRed))$};
    \node [right] at (7.2,-6) {$ \wedge ( \mustGoYellowRed \Leftrightarrow (\getReady,\yellowRed))$};

    \node[initial,state,draw=white,fill=white] (green2) at (1,0) {};
    \trafficlight{white}{white}{green}{1}{0}
    
    \node[state,draw=white,fill=white] (red2) at (5,0) {};
    \trafficlight{red}{white}{white}{5}{0}
    
    \node[state,draw=white,fill=white] (yellow2) at (3,-2) {};
    \trafficlight{white}{yellow}{white}{3}{-2}
    
    \node[state,draw=white,fill=white] (yellowred2) at (3,2) {};
    \trafficlight{red}{yellow}{white}{3}{2}
    
    \path (red2) edge [bend right=10] node [above] {\go} (green2);
    \path (green2) edge [bend right=10] node [below] {\stop} (red2);
    \path (red2.north) edge [bend right] node [above,sloped] {\getReady} (yellowred2.east);
    \path (yellowred2.west) edge [bend right] node [above,sloped] {\go} (green2.north);
    \path (green2.south) edge [bend right] node [below,sloped] {\prepareToStop} (yellow2.west);
    \path (yellow2.east) edge [bend right] node [below,sloped] {\stop} (red2.south);
    
    \node [right] at (.3,3.2) {Parameters: $P=\{\mustGoYellowBoth\}$};
    \node [right] at (-1.7,-3.5) {Obligation function:};
    \node [right] at (-1.7,-6.5) {$\Phi(\yellowRed) =  (\go,\green)$};
    \node [right] at (-1.7,-4) {$\Phi(\green) = ((\stop,\red) \oplus (\prepareToStop,\yellow))$};
    \node [right] at (-1.1,-4.5) {$\wedge ( \mustGoYellowBoth \Leftrightarrow (\prepareToStop,\yellow)) $};
    \node [right] at (-1.7,-5) {$\Phi(\yellow) =  (\stop,\red)$};
    \node [right] at (-1.7,-5.5) {$ \Phi(\red) = ((\go,\green) \oplus (\getReady,\yellowRed))$};
    \node [right] at (-1.1,-6) {$ \wedge ( \mustGoYellowBoth \Leftrightarrow (\getReady,\yellowRed))$};

    \node[initial,state,draw=white,fill=white] (green0) at (-6,3) {};
    \trafficlight{white}{white}{green}{-6}{3}
    
    \node[state,draw=white,fill=white] (red0) at (-3,3) {};
    \trafficlight{red}{white}{white}{-3}{3}

    \path (red0) edge [bend right=10] node [above] {\go} (green0);
    \path (green0) edge [bend right=10] node [below] {\stop} (red0);

    \node[initial,state,draw=white,fill=white] (green) at (-6,-2) {};
    \trafficlight{white}{white}{green}{-6}{-2}
    
    \node[state,draw=white,fill=white] (red) at (-3,-2) {};
    \trafficlight{red}{white}{white}{-3}{-2}
    
    \node[state,draw=white,fill=white] (yellow) at (-4.5,-3.5) {};
    \trafficlight{white}{yellow}{white}{-4.5}{-3.5}
    
    \node[state,draw=white,fill=white] (yellowred) at (-4.5,-.5) {};
    \trafficlight{red}{yellow}{white}{-4.5}{-.5}

    \path (red.north) edge [bend right] node [above,sloped] {\getReady} (yellowred.east);
    \path (yellowred.west) edge [bend right] node [above,sloped] {\go} (green.north);
    \path (green.south) edge [bend right] node [below,sloped] {\prepareToStop} (yellow.west);
    \path (yellow.east) edge [bend right] node [below,sloped] {\stop} (red.south);

\node[font=\LARGE] at (6.4,0) {$\mr$} ;
\node[font=\LARGE, rotate=-20] at (-0.6,1.9) {$\mr$} ;    
\node[font=\LARGE, rotate=20] at (-0.6,-1.5) {$\mr$} ;
\end{tikzpicture}  }
\end{center}
\end{example}

The original version of~\cite{DBLP:conf/atva/BenesKLMS11} requires for $s_0\mr t_0$ to hold that there be a \emph{fixed} 
$R \subseteq S_1 \times S_2$ such that for every $\mu \subseteq P_1$ there exists $\nu \subseteq P_2$
satisfying for each $(s,t) \in R$ 
\begin{align*}  
 \forall M \in \Tran_{\mu}(s) : \exists N \in \Tran_{\nu}(t) : 
        & \ \ \forall (a,s') \in M : \exists (a,t') \in N : (s',t') \in R \ \ \wedge \\
                &\ \ \forall (a,t') \in N : \exists (a,s') \in M : (s',t') \in R\ .
\end{align*}
Clearly, the original definition is stronger:
For any two PMTS states, if $s_0\mr t_0$ holds according to~\cite{DBLP:conf/atva/BenesKLMS11} it also holds according to Definition~\ref{def:mr-pmts}. Indeed, the relation for any sets of parameters can be chosen to be the fixed relation $R$. On the other hand, the opposite does not hold.

\begin{example}\label{ex:pmts-defs-diff}
Consider the PMTS on the left with parameter set $\{p\}$ and obligation $\Phi(s_0)=(a,s_1),\ \Phi(s_1)=(b,s_2)\Leftrightarrow p, \Phi(s_2)=\true$ and the PMTS on the right with parameter set $\{q\}$ and obligation $\Phi(t_0)=\big((a,t_1)\Leftrightarrow q\big)\wedge\big((a,t_1')\Leftrightarrow \neg q\big),\ \Phi(t_1)=(a,t_2), \Phi(t_2)=\Phi(t_1')=\true$. On the one hand, according to our definition $s_0\mr t_0$. We intuitively agree it should be the case (and note they also have the same set of implementations). On the other hand, the original definition does not allow to conclude modal refinement between $s_0$ and $t_0$. The reason is that depending on the value of $p$, $s_1$ is put in the relation either with $t_1$ (for $p$ being true and thus choosing $q$ true, too) or with $t_1'$ (for $p$ being false and thus choosing $q$ false, too). In contrast to the original definition, our definition allows us to pick different relations for different parameter valuations.

\vspace*{-0.5em}

\begin{center}
\begin{tikzpicture}[->,>=stealth',initial text=,xscale=1.6]
\node[state,initial] (s) at(0,0) {$s_0$};
\node[state] (s1) at(1,0) {$s_1$};
\node[state] (s2) at(2,0) {$s_2$};

    \path (s) edge node [above] {$a$} (s1);
    \path (s1) edge node [above] {$a$} (s2);

\node[state,initial] (t) at(4,0) {$t_0$};
\node[state] (t1) at(5,0.5) {$t_1$};
\node[state] (t2) at(6,0.5) {$t_2$};
\node[state] (t3) at(5,-0.5) {$t_1'$};

    \path (t) edge node [above] {$a$} (t1) (t1) edge node [above] {$a$} (t2);
    \path (t) edge node [below] {$a$} (t3);
\end{tikzpicture} 
\end{center}
\end{example}
\vspace*{-0.5em}

We propose our modification of the definition since it is more intuitive and for all considered fragments of PMTS has the same complexity as the original one. Note that both definitions coincide on BMTS. Further, on MTS they coincide with Definition~\ref{def:mr-mts} and on labelled transition systems with bisimulation.




\section{Modal Refinement Checking}\label{sec:mr}

In this section, we show how to solve the modal refinement problem on BMTS and PMTS using QBF solvers. Although modal refinement is $\Pi_2$-complete (the second level of the polynomial hierarchy) on BMTS and $\Pi_4$-complete on PMTS (see~\cite{DBLP:conf/atva/BenesKLMS11}), this way we obtain a solution method that is practically fast. We have implemented the approach and document its scalability on experimental results.

As mentioned, in order to decide whether modal refinement holds between two states, a reduction to a quantified boolean formula will be used. First, we recall the QBF decision problems.

\begin{definition}
[${QBF_n^Q}$] Let $Ap$ be a set of atomic propositions, which is partitioned into $n$ sets with $Ap = \bigcup_{i=0}^{n}X_i$, and $\phi \in \mathcal{B}(Ap)$ a boolean formula over this set of atomic propositions. Let $Q \in \{ \forall , \exists \} $ be a quantifier and $\overline{\phantom{X}} : \{ \forall \mapsto \exists, \exists \mapsto \forall \}$ a function. Then a formula
\[
Q X_1 \overline{Q} X_2 Q X_3 \dots \widetilde{Q} X_n \phi \quad \text{with $\widetilde{Q}=\begin{cases}
Q& \text{if $n$ is odd}\\
\overline{Q}& \text{if $n$ is even}
\end{cases}$}
\]
is an instance of $QBF_n^Q$ if it is satisfiable. 
\end{definition}
Satisfiability means that if e.g. $Q = \exists$ there is some partial valuation for the atomic propositions in $X_1$, such that for all partial valuations for the elements of $X_2$, there is another partial valuation for the propositions of $X_3$ and so on up to $X_n$, such that $\phi$ is satisfied by the union of all partial valuations.
It is well known 
that these problems are complete for the polynomial hierarchy:
For each $i \geq 1$,  $QBF^\exists_i$ is $\Sigma_i$-complete and $QBF^\forall_i$ is $\Pi_i$-complete.

\subsection{Construction for BMTS}

Due to the completeness of QBF problems and the results of~\cite{DBLP:conf/atva/BenesKLMS11}, it is possible to polynomially reduce modal refinement on BMTS to $QBF^\forall_2$. However, we would then have to perform a fixpoint computation to compute the refinement relation causing numerous invocations of the external QBF solver. Hence it is faster to guess the relation and thus reduce the modal refinement only to $QBF^\exists_3$. 

Let $s \in S_1$ and $t \in S_2$ be processes of two arbitrary {BMTS}s $\xmts_1 = (S_1, T_1, \emptyset, \Phi_1)$ and $\xmts_2 = (S_2, T_2, \emptyset, \Phi_2)$. Furthermore let 
\[ Ap = \underbrace{(S_1 \times S_2)}_{X_R} \uplus \underbrace{T_1}_{X_{T1}} \uplus \underbrace{(S_1 \times T_2)}_{X_{T2}}\] 
be a set of atomic propositions.
The intended meaning is that $\tuple{u,v} \in X_R$ is assigned $\true$ if and only if  it is also contained in the modal refinement relation $R$. Further, $X_{T1}$ and $X_{T2}$ are used to talk about the transitions. The prefix $S_1$ is attached to the set $T_2$ because $N \in \Tran(t)$ with $t \in S_2$ must be chosen independently for different states of $S_1$. This trick enables us later to pull up the $\exists$ quantification in the formula.

We now construct a formula $\Psi_{s,t} \in \mathcal{B}(Ap)$ satisfying 
\begin{equation}
\label{reduction}
s \leq_m t \quad \mathit{iff} \quad \exists X_R \forall X_{T1} \exists X_{T2} \Psi_{s,t} \in QBF^\exists_3
\end{equation}

To this end, we shall use a macro $\psi_{u,v}$ capturing the condition which has to be satisfied by any element $(u,v) \in R$. Furthermore, we ensure that $(s,t)$ is assigned \textbf{tt} by every satisfying assignment for the formula by placing it directly in the conjunction:

\begin{equation}
\label{reduction:l1}
\Psi_{s,t} = \tuple{s,t} \wedge {\bigwedge_{\tuple{u,v}\in X_R}}{\bigl(\tuple{u,v}\Rightarrow\psi_{u,v}\bigr)}
\end{equation}

It remains to define the macro $\psi_{u,v}$. We start with the modal refinement condition as a blueprint:

\begin{align*}  
 \forall M \in \Tran(u) : \exists N \in \Tran(v) : 
        & \ \ \forall (a,u') \in M : \exists (a,v') \in N : (u',v') \in R \ \ \wedge \\
                &\ \ \forall (a,v') \in N : \exists (a,u') \in M : (u',v') \in R\ .
\end{align*}

As $M$ and $N$ are subsets of $T_1(u)$ and $T_2(v)$, respectively, and are finite, the inner quantifiers can be expanded causing only a polynomial growth of the formula size (see Appendix~\ref{app:mr}). Further, $\Tran$ sets are replaced by the original definition and the outer quantifiers are moved in front of $\Psi_{s,t}$. As the state obligations are defined over a different set of atomic propositions ($\Phi(v) \in \mathcal{B} ((\Sigma \times S) \cup P) \not \subseteq \mathcal{B}(Ap)$), a family of mapping functions $\pi_{p}$ is introduced.

\begin{equation}
\begin{split}
\pi_{p}: \mathcal{B}(\Sigma \times S) & \rightarrow \mathcal{B}(Ap)\\
   \mathbf{tt}  & \mapsto \mathbf{tt} \\
   \tuple{a,x}  & \mapsto \tuple{p,a,x} \quad \; \mathit{with} \; a \in \Sigma, \, x \in S \\
   \neg \varphi & \mapsto \neg \;\pi_{p}(\varphi) \\
   \varphi_1 \wedge \varphi_2 & \mapsto \pi_{p}(\varphi_1) \wedge \pi_{p}(\varphi_2) \\
   \varphi_1 \vee \varphi_2 & \mapsto \pi_{p}(\varphi_1) \vee \pi_{p}(\varphi_2)
\end{split}
\end{equation}

A applying these steps to the blueprint yields the following result:

\begin{equation}
\label{reduction:l2}
\psi_{u,v}=\pi_{u}\left(\Phi_1\left(u\right)\right) \Rightarrow \pi_{u,v}\left(\Phi_2\left(v\right)\right)\wedge\varphi_{u,v}
\end{equation}
\begin{equation}
\label{reduction:l3}
\begin{split}
\varphi_{u,v} = & \; \bigwedge_{\mathclap{{\substack{u^* \in X_{T1}\\u^* = (u, a, u')}}}} \;
\bigl({u^*} \Rightarrow \; \bigvee_{\mathclap{{\substack{v^* \in X_{T2}\\v^* = (u, v, a, v')}}}} \;
\left({v^*} \wedge \left({u',v'}\right) \right) \bigr) \\  
\wedge & \; \bigwedge_{\mathclap{{\substack{v^* \in X_{T2}\\v^* = (u, v, a, v')}}}} \;
\bigl({v^*} \Rightarrow \; \bigvee_{\mathclap{{\substack{u^* \in X_{T1}\\u^* = (u, a, u')}}}} \;
\left({u^*} \wedge \left({u',v'}\right) \right) \bigr) \\
\end{split}
\end{equation}

\begin{theorem}\label{thm:mr-red}
For states $s,t$ of a BMTS, we have
$$s \leq_m t \quad \mathit{iff} \quad \exists X_R \forall X_{T1} \exists X_{T2} \Psi_{s,t} \in QBF^\exists_3$$ 
\end{theorem}

Due to space constraints, the technical proof is moved to Appendix~\ref{app:mr}.

\subsection{Construction for PMTS}

We now reduce the modal refinement on PMTS to $QBF^\forall_4$, which now corresponds directly to the complexity established in~\cite{DBLP:conf/atva/BenesKLMS11}. Nevertheless, due to the first existential quantification in $\forall \exists \forall \exists$ alternation sequence, we can still guess the refinement relation using the QBF solver rather than compute the lengthy fixpoint computation.

In the PMTS case, we have to find for all parameter valuations for the system of $s$ a valuation for the system of $t$, such that there exists a modal refinement relation containing $(s,t)$. We simply choose universally a valuation for the parameters of the left system (the underlying system of $s$) and then existentially for the right system (the underlying system of $t$). Prior to checking modal refinement, the valuations are fixed, so the {PMTS} becomes a {BMTS}. This is accomplished by extending $Ap$ with $P_1$ and $P_2$ and adding the necessary quantifiers to the formula. Thus we obtain the following:

\begin{theorem}\label{thm:mr-red}
For states $s,t$ of a PMTS, we have
$$s \leq_m t \quad \mathit{iff} \quad \forall P_1 \exists P_2 \exists X_R \forall X_{T1} \exists X_{T2} \Psi_{s,t}  \in QBF^\forall_4$$ 
\end{theorem}


\subsection{Experimental Results}

We now show how our method performs in practice. We implemented the reduction and linked it to the QBF solver Quantor. In order to evaluate whether our solution scales, we generate random samples of MTS, disjunctive MTS, Boolean MTS and parametric MTS with different numbers of parameters (as displayed in tables below in parenthesis). For each type of system and the number of reachable states (25 to 200 as displayed in columns), we generate several pairs of systems and compute the average time to check modal refinement between them. 

We show several sets of experiments. In Table~\ref{table:mr}, we consider (1) systems with alphabet of size 2 and all states with branching degree 2, and (2) systems with alphabet of size 10 and all states with branching degree 10. Further, in Table~\ref{table:gl}, we consider systems with alphabet of size 2 and all states with branching degree 5. Here we first consider the systems as above, i.e.~with edges generated randomly so that they create a tree and with some additional ``noise'' edges thus making the branching degree constant. Second, we consider systems where we have different ``clusters'', each of which is interconnected with many edges. Each of these clusters has a couple of ``interface'' states, which are used to connect to other
clusters. We use this class of systems to model system descriptions with more organic structure. 

The entries in the tables are average running times in seconds. The standard deviation in our experiments was around 30-60\%.  Each star denotes that on one of five experiments, the QBF solver Quantor timed out after one minute. The experiments were run on Intel Core 2 Duo CPU P9600 \@ 2.66GHz x 2 with 3.8 GB RAM using Java 1.7. For more details and more experiments, see~\texttt{http://www.model.in.tum.de/\textasciitilde kretinsk/ictac13.html}.

\vspace*{-1em}
\begin{table}
\caption{Experimental results: systems over alphabet of size 2 with branching degree 2 in the upper part, and systems over alphabet of size 10 with branching degree 10 in the lower part}
\label{table:mr}
\begin{minipage}{2\textwidth}
    \begin{tabular}{|l|rrrrrrrr|}
        \hline
 &       	~~25~~	&	~~50~~	&	~~75~~	&	~~100~~	&	~~125~~	&	~~150~~	&	~~175~~	&	~~200~~	\\
\hline
MTS	&	0.03	&	0.15	&	0.29	&	0.86	&	0.87	&	0.96	&	1.88	&	2.48	\\
DMTS&	0.04	&	0.22	&	0.39	&	0.91	&	1.13	&	1.34	&	2.61	&	3.19	\\
BMTS&	0.03	&	0.15	&	0.30	&	0.62	&	0.83	&	0.87	&	1.61	&	2.17	 \\
PMTS(1)	&0.03	&	0.20	&	0.37	&	0.84	&	0.97	&	1.23	&	2.44	&	3.15	\\
PMTS(5)	&0.04	&	0.22	&	0.42	&	0.91	&	1.26	&	1.59	&	2.83	&	3.66	\\
\hline
MTS	&	0.18	&	0.84	&	2.12	&	3.88	&	5.63	&	7.64	&	10.30	&	14.18	\\
DMTS&	0.44	&	2.23	&	5.31	&	8.59	&	10.13	&	14.14	&	13.96	&	66.92	\\
BMTS&	0.21	&	1.08	&	2.65	&	4.58	&	6.70	&	9.63	&	12.44	&	17.06	\\
PMTS(1)&	0.26	&	1.12	&	2.74	&	4.57	&	7.58	&	10.31	&	11.26	&	16.41 	\\
PMTS(5)&	0.25	&	1.17	&	2.94	&	6.36	&	7.80	&	10.01	&	11.90	&	36.51	\\
        \hline
    \end{tabular}
~
\begin{minipage}{\textwidth}
\includegraphics[scale=0.4]{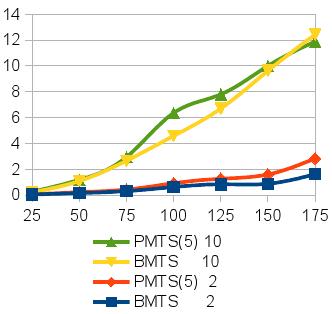}
\end{minipage}
\end{minipage}
\end{table}


\vspace*{-1em}

\begin{table}
\caption{Experimental results: systems over alphabet of size 2 with branching degree 5; systems with random structure in the upper part, and systems with organic structure in the lower part}
\label{table:gl}
    \begin{tabular}{|l|rrrrrrrr|}
        \hline
		&	~~25~~		&	~~50~~		&	~~75~~		&	~~100~~		&	~~125~~		&	~~150~~		&	~~175~~		&	~~200~~		\\\hline
PMTS	(1)	&	0.34		&	2.04		&	5.38		&	8.81		&	11.78		&	17.41		&	27.33		&	58.06		\\
PMTS	(5)	&	0.29		&	1.83		&	*5.19		&	12.79		&	15.71		&	26.60	&	*35.30		&	89.25		\\
PMTS	(10)	&	*0.43		&	1.36		&	6.70		&	13.66		&	*18.27		&	*21.10		&	51.67		&	232.83		\\
\hline
PMTS	(1)	&	0.05		&	0.14		&	0.18		&	0.30		&	3.40		&	0.73		&	0.85		&	0.96		\\
PMTS	(5)	&	0.02		&	0.04		&	0.23		&	0.70		&	0.58		&	0.39		&	1.13		&	*2.35		\\
PMTS	(10)	&	0.02		&	0.10		&	0.16		&	*0.16		&	*0.29		&	1.55		&	0.97		&	1.13		\\
        \hline
    \end{tabular}
\end{table}

\vspace*{-1em}

On the one hand, observe that the number of parameters does not play any major r\^ole in the running time. The running times on PMTS with 5 or even more parameters are very close to BMTS, i.e.~PMTS with zero parameters, as can be seen in the graph. Therefore, the greatest theoretical complexity threat---the number of parameters allowing in general only for searching all exponentially many combinations---is in practice eliminated by the use of QBF solvers. 

On the other hand, observe that the running time is more affected by the level of non-determinism. For branching degree 10 over 10-letter alphabet there, there are more likely to be more outgoing transitions under the same letter than in the case with branching degree 2 over 2-letter alphabet, but still less than for branching degree 5 over 2-letter alphabet. However, the level of non-determinism is often quite low~\cite{DBLP:journals/tcs/BenesKLS09}, hence this dependency does not pose so serious problem in practice. Further, even this most difficult setting with high level of non-determinism allows for fast analysis if systems with natural organic structure are considered, cf.~upper and lower part of Table~\ref{table:gl}.

A more serious problem stems from our use of Java. With sizes around 200, the running times often get considerably longer, see the tables. Here the memory management and the garbage collection take their toll. However, this problem should diminish in a garbage-collection-free setting.




\section{Thorough Refinement Checking}\label{sec:tr}

While modal refinement has been defined syntactically, there is also a corresponding notion defined semantically. The semantics of a state $s$ of a PMTS is the set of its implementations $\semantics{s}:=\{i\mid \text{$i$ is an implementation and } i\mr s\}$.

\begin{definition}[Thorough Refinement]
For states $s_0$ and $t_0$ of PMTS, we say that $s_0$ \emph{thoroughly refines} $t_0$, written $s_0\tr t_0$, if $\semantics{s_0}\subseteq\semantics{t_0}$.
\end{definition}

\subsection{Transforming PMTS to BMTS and DMTS}

The thorough refinement problem is EXPTIME-complete for MTS~\cite{DBLP:journals/iandc/BenesKLS12} and also for DMTS~\cite{DBLP:conf/atva/BenesCK11} (for proof, see~\cite{DBLP:conf/atva/BenesCK11-techrep}). 
First, we show how to transform PMTS to BMTS and DMTS and thus reduce our problems to the already solved one.

For a PMTS, we define a system where we can use any valuation of the parameters:

\begin{definition}
For a PMTS $\xmts=(S,T,P,\Phi)$ with initial state $s_0$, we define a BMTS called \emph{de-parameterization} $\xmts^B=(\{s_0^B\}\cup S\times 2^P,T',\emptyset,\Phi')$ with initial state $s_0^B$ and 
\begin{itemize}
 \item $T=\{(s_0^B,a,(s,\nu))\mid (s_0,a,s)\in T,\nu\subseteq P\}\cup\{((s,\nu),a,(s',\nu)\mid (s,a,s')\in T\}$,
 \item $\Phi'(s_0^B)=\displaystyle\bigoplus_{\nu\subseteq P}\Phi(s_0)[\true/p \text{ for } p\in\nu,\false/p \text{ for } p\notin\nu,(s,\nu)/s]$,
 \item $\Phi'\big((s,\nu)\big)\ =\ \ \Phi(s)[\true/p \text{ for } p\in\nu,\false/p \text{ for } p\notin\nu,(s,\nu)/s]$.
\end{itemize}
\end{definition}

The de-parameterization is a BMTS having exactly all the implementations of the PMTS and only one (trivial) valuation.

\begin{proposition}
Let $s_0$ be a PMTS state. Then $\semantics{s_0}=\semantics{s_0^B}$ and $s_0\mr s_0^B$.  
\end{proposition}
\begin{proof}
For any parameter valuation $\nu$ we match it with $\emptyset$ and the modal refinement is achieved in the copy with $\nu$ fixed in the second component. Clearly, any implementation of $s_0^B$ corresponds to a particular parameter valuation and thus also to an implementation of $s_0$.
\qed\end{proof}

\begin{remark}
The price we have to pay is a blowup exponential in $|P|$. This is, however, inevitable. Indeed, consider a PMTS $(\{s_0,s_1,s_2\},\{(s_0,p,s_1),(s_1,p,s_2)\mid p\in P\},P,\{s_0,s_1\mapsto \bigwedge_{p\in P}(p,s)\Leftrightarrow p, s_2\mapsto\true\})$. Then in every equivalent BMTS we need to remember the transitions of the first step so that we can repeat exactly these in the following step. Since there are exponentially many possibilities, the result follows.
\end{remark}

Further, similarly to Boolean formulae with states in~\cite{bfs}, we can transform every BMTS to a DMTS.

\begin{definition}
For a BMTS $\xmts=(S,T,\emptyset,\Phi)$ with initial state $s_0$, we define a DMTS called \emph{de-negation} $\xmts^D=(S',T',\emptyset,\Phi')$ 
\begin{itemize}
\item $S'=\{ M\in \Tran( s)\mid s\in S\}$,
\item $\Phi'(M)=\bigwedge_{( a, s')\in  M} \bigvee_{M'\in \Tran( s')}( a, M')$,
\end{itemize}
and $T'$ minimal such that for each $M\in S'$ and each occurrence of $(a,M')$ in $\Phi(M)$, we have $(M,a,M')\in T'$.
\end{definition}

However, this DMTS needs to have more initial states in order to be equivalent to the original BMTS:

\begin{lemma}
For a state $s_0$ of a BMTS, $\semantics{s_0}=\bigcup_{M\in\Tran(s_0)}\semantics{M}$ (where $M$ are taken as states of the de-negation).
\end{lemma}

Note that both transformations are exponential. The first one in $|P|$ and the second one in the branching degree. Therefore, their composition is still only singly exponential yielding a state space where each state has two components: a valuation of original parameters and $\Tran$ of the original state under this valuation.

\begin{theorem}
Thorough refinement on PMTS is in 2-EXPTIME.
\end{theorem}
\begin{proof}
Recall that thorough refinement on DMTS is in EXPTIME. Further, note that we have reduced the PMTS and BMTS thorough refinement problems to the one on DMTS with more initial states. However, this does not pose a problem. Indeed, let $s_0$ and $t_0$ be states of a BMTS. We want to check whether $s_0\tr t_0$. According to~\cite{DBLP:conf/atva/BenesCK11-techrep} where DMTS only have one initial state, we only need to check whether for each $M\in\Tran(s_0)$ we have $(M,\Tran(t_0))\notin Avoid$, which can clearly still be done in exponential time.
\qed\end{proof}

\subsection{Direct algorithm}

We now extend the approach for MTS and DMTS to the {BMTS} case. Before proceeding, one needs to prune all inconsistent states, i.e. those with unsatisfiable obligation. This is standard and the details can be found in Appendix~\ref{app:tr}. 

We define a set $Avoid$, which contains pairs consisting of one process and one set of processes. A pair is contained in the relation if there exists an implementation refining the single process, but none of the other processes. This approach is very similar to~\cite{DBLP:journals/iandc/BenesKLS12}, but the rules for generating $Avoid$ are much more complex.

\begin{definition} \emph{(Avoid)}
Let $(S,T,\emptyset,\Phi)$ be a globally consistent BMTS over the action alphabet $\Sigma$. The set of \emph{avoiding} states of the form $(s, \mathcal{T})$, where $s \in S$ and $\mathcal{T} \subseteq S$, is the smallest set $Avoid$ such that $(s, \mathcal{T}) \in Avoid$ whenever $\mathcal{T} = \emptyset$ or there exists an admissible set of transitions $M \in \Tran(s)$ and sets $later_{a,u,f} \subseteq S$ for every $a \in \Sigma$, $u\in S$, $f\in\bigcup_{t\in\mathcal T}\Tran(t)$ such that
\begin{align*}
\forall t \in \mathcal{T}: \;& \forall N_t \in \Tran(t): \; \exists a \in \Sigma: \\
           & \exists t_a \in N_t(a): \forall s_a \in M(a): \forall f\in\bigcup_{t\in\mathcal T}\Tran(t): t_a \in later_{a,s_a,f} \\
\vee \quad & \exists s_a \in M(a): \forall t_a \in N_t(a): t_a \in later_{a,s_a,N_t}
\end{align*}
and
\begin{equation*}
\forall f\in\bigcup_{t\in\mathcal T}\Tran(t): \forall (a, s_a) \in M: \left(s_a, later_{a,s_a,f}\right) \in Avoid
\end{equation*}
hold.
\end{definition}

\begin{lemma}\label{lem:tr}
Given processes $s, t_1, t_2 \dots t_n$ of some finite, global-consistent {BMTS}, there exists an implementation $I$ such that $I \mr s$ and $I \not \mr t_i$ for all $i \in [1, n]$ if $\left(s, \{t_1, t_2 \dots t_n\}\right) \in Avoid$.
\end{lemma}

\begin{theorem}
Thorough refinement checking on BMTS is in NEXPTIME.
\end{theorem}

\begin{proof}
For deciding $s \tr t$ the $Avoid$ relation has to be computed, whose size grows exponentially with the size of the underlying system. Moreover, in each step of adding a new element is added to $Avoid$, the sets $later_{a,s,f}$ need to be guessed. 
\qed\end{proof}




\section{Thorough vs.\ Modal Refinement}\label{sec:mrtr}

In this section, we discuss the relationship of the two refinements. Some proofs are moved to Appendix~\ref{app:mrtr}.
Firstly, the modal refinement is a sound approximation to the thorough refinement.

\begin{proposition}\label{prop:trans}
Let $s_0$ and $t_0$ be states of PMTS. If $s_0\mr t_0$ then also $s_0\tr t_0$.
\end{proposition}
\begin{proof}
For any $i\in\semantics{s_0}$, we have $i\mr s_0$ and due to transitivity of $\mr$, $i\mr s_0\mr t_0$ implies $i\mr t_0$, hence $i\in\semantics{t_0}$.
\qed\end{proof}

The converse fails already for MTS as shown in the following classical example (\cite{DBLP:journals/tcs/BenesKLS09}) where $s_0 \tr t_0$, but $s_0 \not\mr t_0$.

\vspace*{-1em}
\begin{center}
\begin{tikzpicture}[->,>=stealth',initial text=,xscale=2]
\tikzset{
    state/.style={
		rectangle,
            rounded corners,
            draw=black,
            minimum height=2em,
            minimum width=2em,
            inner sep=4pt,
outer sep=2pt,
            text centered,
            }
}

\node[state] (s) at(0,0) {$s_0$};
\node[state] (s1) at(1,0) {$s_1$};
\node[state] (s2) at(2,0) {$s_2$};

    \path (s) edge[dashed] node [above] {$a$} (s1);
    \path (s1) edge[dashed] node [above] {$a$} (s2);

\node[state] (t) at(3,0) {$t_0$};
\node[state] (t1) at(4,0.5) {$t_1$};
\node[state] (t2) at(5,0.5) {$t_2$};
\node[state] (t3) at(4,-0.5) {$t_1'$};

    \path (t) edge[dashed] node [above] {$a$} (t1) (t1) edge node [above] {$a$} (t2);
    \path (t) edge[dashed] node [below] {$a$} (t3);
\end{tikzpicture}
\end{center}
\vspace*{-1em}

However, provided the refined MTS is deterministic, the approximation is also complete~\cite{DBLP:journals/tcs/BenesKLS09}. This holds also for BMTS. This is very useful as deterministic system often appear in practice~\cite{DBLP:journals/tcs/BenesKLS09} and checking modal refinement is computationally easier than the thorough refinement. Formally, we say that a PMTS $(S,T,P,\Phi)$ is deterministic if for every $(s,a,t),(s,a,t')\in T$ we have $t=t'$.

\begin{proposition}\label{prop:rightdet}
Let $s_0$ be a PMTS state and $t_0$ a deterministic BMTS state. If $s_0\tr t_0$ then also $s_0\mr t_0$.
\end{proposition}

However, the completeness fails if the refined system is deterministic but with parameters:

\begin{example}
Consider a BMTS 
$(\{s_0,s_1\},\{s_0,a,s_1\},\emptyset,\{s_0\mapsto \true,s_1\mapsto\true\})$ and a deterministic PMTS $(\{t_0,t_1\},\{(t_0,a,t_1)\},\{p\},\{t_0\mapsto a\Leftrightarrow p,t\mapsto\true\})$ below. Obviously $\semantics{s_0}=\semantics{t_0}$ contains the implementations with no transitions or one step $a$-transitions. Although $s_0\tr t_0$, we do not have $s_0\mr t_0$ as we cannot match with any valuation of $p$.

\begin{center}
\begin{tikzpicture}[->,>=stealth',initial text=,xscale=2]
\tikzset{
    state/.style={
		rectangle,
            rounded corners,
            draw=black,
            minimum height=2em,
            minimum width=2em,
            inner sep=4pt,outer sep=2pt,
            text centered,
            }
}

\node[state] (s) at(0,0) {$s_0$};
\node[state] (s1) at(1,0) {$s_1$};

    \path (s) edge node [above] {$a$} (s1);

\node[state] (t) at(2,0) {$t_0$};
\node[state] (t1) at(3,0) {$t_1$};

    \path (t) edge node [above] {$a$} node[below=0.5]{$\Phi(t_0)=\ a\Leftrightarrow p$~~~~~~~~} (t1);

\node[state] (u) at(4,0) {$t_0^B$};
\node[state] (u1) at(5,0) {$(t_1,p=\true)$};

    \path (u) edge node [above] {$a$} node[below=0.5]{~~~~~$\Phi(t_0^B)=\ (a\Leftrightarrow \true) \vee (a\Leftrightarrow \true)$} (u1);
\end{tikzpicture}

\end{center}

\end{example}

\begin{corollary}
There is a state $s_0$ of a PMTS and a state $t_0$ of a deterministic PMTS such that $s_0\tr t_0$ but $s_0\not\mr t_0$. 
\end{corollary}

In the previous example, we lacked the option to match a system with different parameter valuations at once.
However, the de-parameterization introduced earlier is non-deterministic even if the original system was deterministic. Hence the modal refinement is not guaranteed to coincide with the thorough refinement.
In \cite{DBLP:journals/tcs/BenesKLS09}, we defined the notion of deterministic hull, the best deterministic overapproximation of a system. The construction on may transitions was the standard powerset construction and a must transition was created if all states of a macrostate had one. Here we extend this notion to PMTS, which allows to over- and under-approximate the thorough refinement by the modal refinement.

\begin{definition}
For a PMTS $\xmts=(S,T,P,\Phi)$ with initial state $s_0$, we define a PMTS called \emph{deterministic hull} $\deth(\xmts)=(2^S,T',P,\Phi')$ with initial state $\deth(s_0):=\{s_0\}$ and 
\begin{itemize}
 \item $T=\{(S,a,S_a\}$ where $S_a$ denotes all $a$-successors of elements of $S$, i.e. $S_a=\{s'\mid \exists s\in S: (s,a,s')\in T)\}$,
 \item $\Phi'(S)=\bigvee_{s\in S}\Phi(s)[(a,S_a)/(a,s) \text{ for every $a,s$}]$.
\end{itemize}
\end{definition}

\begin{proposition}\label{prop:hull-sound}
For a PMTS state $s_0$, $\deth(s_0)$ is deterministic and $s_0\mr \deth(s_0)$. 
\end{proposition}
We now show the minimality of the deterministic hull.

\begin{proposition}\label{prop:hull-complete}
Let $s_0$ be a PMTS state. Then
\begin{itemize}
 \item for every deterministic PMTS state $t_0$, if $s_0\mr t_0$ then $\deth(s_0)\mr t_0$;
 \item for every deterministic BMTS state $t_0$, if $s_0\tr t_0$ then $\deth(s_0)\mr t_0$.
\end{itemize}
\end{proposition}
The next transformation allows for removing the parameters without introducing non-determinism.

\begin{definition}
For a PMTS $\xmts=(S,T,P,\Phi)$ with initial state $s_0$, we define a BMTS called \emph{parameter-free hull} $\parh(\xmts)=(S,T,\emptyset,\Phi')$ with initial state $\parh(s_0):=s_0$ and 
$$\Phi'(s)=\bigvee_{\nu\subseteq P}\Phi(s)[\true/p \text{ for } p\in\nu,\false/p \text{ for } p\notin\nu]$$
\end{definition}

\begin{lemma}\label{prop:phull-sound}
For a PMTS state $s_0$, $s_0\mr s_0^B \mr \parh(s_0)$. 
\end{lemma}
The parameter-free deterministic hull now plays the r\^ole of the deterministic hull for MTS.

\begin{corollary}
For PMTS states $s_0$ and $t_0$, if $s_0\tr t_0$ then $s_0\mr \parh(\deth(t_0))$.
\end{corollary}
\begin{proof}
Since $s_0\tr t_0$, we also have $s_0\tr \deth(t_0)$ by Propositions~\ref{prop:hull-sound} and~\ref{prop:trans}. Therefore, $s_0\tr \parh(\deth(t_0))$ by Proposition~\ref{prop:phull-sound} and thus $s_0\mr \parh(\deth(t_0))$ by Proposition~\ref{prop:rightdet}.
\qed\end{proof}




\section{Conclusions}\label{sec:concl}

\vspace*{-0.7em}

We have investigated both modal and thorough refinement on Boolean and parametric extension of modal transition systems. Apart from results summarized in the table below, we have shown a practical way to compute modal refinement and use it for approximating thorough refinement. Closing the complexity gap for thorough refinement, i.e.~obtaining matching lower bounds or improving our algorithm remains as an open question.

\vspace*{-.6em}
\begin{center}
\begin{tabular}{|c|ccc|}
\hline
& MTS&BMTS&PMTS\\\hline
$\tr \ \ \in$ & EXPTIME & {~~ NEXPTIME} & {~~ 2-EXPTIME}\\
refined system deterministic & ${\mr\ =\ \tr}$& $\mr\ =\ \tr$& $\mr\ \neq\ \tr$\\\hline
\end{tabular}
\end{center}

\vspace*{-1.7em}

%
%
%
%
%
%
%
%

\bibliographystyle{alpha}
\bibliography{refs}

\newpage

\appendix
\section*{Appendix: Proofs}

\section{Modal Refinement Checking: Proof of Theorem~\ref{thm:mr-red}}\label{app:mr}

Before proving the soundness and the correctness of the construction for BMTS, a lemma is introduced to simplify this proof.

\begin{lemma}
\label{reduction:lemma}
Let be $\tuple{s,t} \in S_1 \times S_2$ a pair of states. Let be $\mathcal{A}_{X_R}$, $\mathcal{A}_{X_{T1}}$ and $\mathcal{A}_{X_{T2}}$ partial valuations for the sets of atomic propositions appearing in their indices. Furthermore let be $R \subseteq S_1 \times S_2$, $M \in \Tran_\emptyset(s)$ and $N \in \Tran_\emptyset(t)$ sets. If $\mathcal{A}_{X_R} = R$, $\mathcal{A}_{X_{T1}} \supseteq \app{\pi_{s}}{M}$ and $\mathcal{A}_{X_{T2}} \supseteq \app{\pi_{s,t}}{N}$ holds, then $\mathcal{A}_{X_R} \cup \mathcal{A}_{X_{T1}} \cup \mathcal{A}_{X_{T2}} \models \varphi_{s,t}$ if and only if 
\begin{align*} 
R \cup M \cup N \models \quad & \forall {(a,s')} \in M: \exists {(a,t')} \in N: {(s', t')} \in R \\
\wedge \quad & \forall {(a,t')} \in N: \exists {(a,s')} \in M: {(s', t')} \in R
\end{align*}
\end{lemma}

\begin{proof} 
We assume the conditions and set $\mathcal{A}_X = \mathcal{A}_{X_R} \cup \mathcal{A}_{X_{T1}} \cup \mathcal{A}_{X_{T2}}$ and $\mathcal{A}_R = R \cup M \cup N$. Additionally, we only consider one half of the conjunction, as the other is proven analogously.

\begin{align*} 
             && \mathcal{A}_R \models \quad & \forall {(a,s')} \in M: \exists {(a,t')} \in N: {(s', t')} \in R \\
\mathit{iff} && \mathcal{A}_R \models \quad & \bigwedge_{\mathclap{{(a,s') \in T_1(s)}}} \;
\bigl({(a,s')} \in M \Rightarrow \; \bigvee_{\mathclap{{(a,t') \in T_2(t)}}} \;
\left({(a,t')} \in N \wedge \left({s',t'} \in R\right) \right) \bigr) \\
\mathit{iff} && \mathcal{A}_X \models \quad & \bigwedge_{\mathclap{{\substack{s^* \in X_{T1}\\s^* = (s, a, s')}}}} \;
\bigl({s^*} \Rightarrow \; \bigvee_{\mathclap{{\substack{t^* \in X_{T2}\\t^* = (s, t, a, t')}}}} \;
\left({t^*} \wedge \left({s',t'}\right) \right) \bigr) \\
\end{align*}

As $M$ and $N$ are finite sets, $\forall$ and $\exists$ quantifiers may simply be expanded. In the second step we simply apply $\pi$ and substitute $\in$ with atomic propositions.
\qed\end{proof}

A relation satisfying the conditions of the definition of the modal refinement is called a \emph{modal refinement relation}.

\subsubsection*{Soundness and Correctness}

'If' part (soundness of the construction). Assume $s \leq_m t$ with the modal refinement relation $R$. As the partial valuation for $X_R$, we set $\mathcal{A}_{X_R} = R$. Furthermore let $\mathcal{A}_{X_{T1}} \subseteq X_{T1}$ be an arbitrary assignment. We now construct an assignment $\mathcal{A}_{X_{T2}}$, such that
\[\mathcal{A} = \mathcal{A}_{X_R} \cup \mathcal{A}_{X_{T1}} \cup \mathcal{A}_{X_{T2}} \models \Psi_{s,t}\]
holds. Without adding anything to $\mathcal{A}_{X_{T2}}$, clearly $\mathcal{A} \models (s,t)$ and $\mathcal{A} \models \tuple{u,v}\Rightarrow\psi_{u,v}$ for all $\tuple{u,v} \in X_R \cap \overline{R}$ hold.

Let now $(u,v) \in R$ be an arbitrary pair of states. If $\mathcal{A} \not \models \app{\pi_u}{\Phi(u)}$, then $\mathcal{A} \models \psi_{u,v}$ and $\mathcal{A} \models (u,v) \Rightarrow \psi_{u,v}$. Hence we assume now $\mathcal{A} \models \app{\pi_u}{\Phi(u)}$. Since $(u,v) \in R$, there exists for all $M \in \Tran_\emptyset(u)$ a set $N$, such that the condition holds, which is included in the assignment $\mathcal{A}_{X_{T2}} \supseteq \app{\pi_{u,v}}{N}$. This can safely be done due to the prefixing and with Lemma \ref{reduction:lemma} we get $\mathcal{A} \models \varphi_{u,v}$ and $\mathcal{A} \models (u,v) \Rightarrow \psi_{u,v}$.

As a valuation $\mathcal{A}$ can be constructed for a fixed modal refinement relation, such that for all subsets of $X_{T1}$ it satisfies the formula, $\exists X_R \forall X_{T1} \exists X_{T2} \Psi_{s,t} \in QBF^\exists_3$ holds.

\bigskip

'Only-If' part (correctness of the construction). We now assume 
\[
\exists X_R \forall X_{T1} \exists X_{T2} \Psi_{s,t} \in QBF^\exists_3 
\]
Then there exists a partial valuation $\mathcal{A}_{X_R} \subseteq \mathcal{A}$ for $X_R$, which satisfies $\Psi_{s,t}$. $R$ is simply constructed by setting $R = \mathcal{A}_{X_R}$. Clearly $(s,t) \in R$. Let now $(u,v) \in R$ be an arbitrary pair of states. As \eqref{reduction:l2} is satisfied for this pair, either $\Phi(u)$ is unsatisfiable and there simply exists no $M \in \Tran_\emptyset(s)$ or for the chosen $M = \pi_{u}^{-1}(\mathcal{A}_{X_{T1}})$ exists a $N = \pi_{u,v}^{-1} (\mathcal{A}_{X_{T2}})$. By Lemma \ref{reduction:lemma} the modal refinement condition holds for this arbitrary pair. Hence $R$ is a modal refinement relation.

\subsubsection*{Polynomial Runtime of the Reduction}

We show that the reduction indeed takes only polynomial time. For this observe that \eqref{reduction:l3} is in $\mathcal{O}(\mid T_1(u) \mid \mid T_2(v) \mid)$. Therefore \eqref{reduction:l2} is in $\mathcal{O}(\mid T_1(u) \mid \mid T_2(v) \mid + \mid \Phi_1(u) \mid + \mid \Phi_2(v) \mid)$. Leading to a total formula size of 
\[\mathcal{O}\Bigl( \mid S_1 \mid \mid S_2 \mid \bigl(\mid T_1 \mid \mid T_2 \mid + \mid \Phi_1 \mid + \mid \Phi_2 \mid\bigr)\Bigr) \]

\section{Thorough Refinement}\label{app:tr}

\subsection{Pruning}

Now the preprocessing is formally introduced. Basically, we prune all the ``inconsistent'' states.

\begin{definition}
[Consistency] A state $s$ of a {BMTS} is called locally consistent if $\Phi(s)$ is satisfiable, otherwise it is called locally inconsistent. 
If all states of a BMTS are locally consistent, the BMTS is called locally consistent. 
A state $s$ of a BMTS is called globally consistent if it has an implementation, i.e.~$\semantics{s}\neq\emptyset$.
\end{definition}

\begin{lemma}
\label{tr:lemma:consistent}
If $(S,T,\emptyset,\Phi)$ is a globally consistent {BMTS}, then for all $s \in S$:
\[ \forall M \in \Tran_\emptyset(s): \exists I \in \llbracket s\rrbracket: T_I(s) = M\]
\end{lemma}

\begin{proof}
Assume the conditions of the lemma. As the {BMTS} is globally consistent, for all $s \in S$ the set $\Tran_\emptyset(s)$ is non-empty. Let now $s \in S$ be an arbitrary state and $M \in \Tran_\emptyset(s)$ an arbitrary set of admissible transitions. We define an implementation $(S_I, T_I, \emptyset, \Phi_I)$ with $S_I = \{t_I \; | \; t \in S \}$, $T_I(s_I) = M$ and for all $t_I \in S \setminus \{s_I\}$ and some $N \in \Tran_\emptyset(t_I)$ we set $T_I(t_I) = N$. As e.g. $R = \{\tuple{t_I,t} \; |\; t \in S\}$ is a suitable modal refinement relation, $s_I \mr s$ holds..
\qed\end{proof}

\begin{corollary}
\label{tr:cor:consistent}
If a state of a BMTS is locally consistent, it is also globally consistent.
\end{corollary}

\begin{proof}
As the system is globally consistent, lemma \ref{tr:lemma:consistent} is applicable. Because $\Tran_\emptyset(s)$ is non-empty for every $s \in S$, there is at least one implementation refining $s$. Thus $\llbracket s \rrbracket \neq \emptyset$.
\qed\end{proof}

As one may have already noted, a locally inconsistent $s \in S$ of some system cannot have any implementation, as $\Tran_\emptyset(s)$ is empty. This is captured by the following lemma.

\begin{lemma}
Removing a locally inconsistent state $s \in S$ from a {BMTS} does not change the semantic $\llbracket t \rrbracket$ of any other process $t \in S \setminus \{s\}$ 
\end{lemma}

\begin{proof}
As the obligation of the state $s$ is unsatisfiable, $\Tran_\emptyset(s)$ is empty. Hence the modal refinement condition is always violated if the left system is locally consistent, which holds for implementations. Therefore, $(u,s) \not \in R$ for any state $u$ of an implementation. Removing the state from the system never affects the modal refinement relation, thus never changes the semantic of any other process of the system.
\qed\end{proof}

However, please note that while removing states from a system does not affect the semantic of other states, it still can make them locally inconsistent. As a preprocessing step, before constructing the $Avoid$ relation, one has to remove all locally inconsistent states until the system becomes globally consistent. If one of the states, for which thorough refinement should be decided, is removed, the decision becomes trivial. If the the left one is inconsistent, the refinement holds. In the other case it does not.

\subsection{Bounded Refinement}

In the course of the proof of Lemma~\ref{lem:tr}, we use a \emph{bounded} version of definition \ref{def:mr-bmts}, which coincides in the limit with the normal definition of modal refinement.

\begin{definition}[\emph{Bounded} Modal Refinement]
\label{def:mr:bounded}
 Let $\mathcal M_1=(S_1, T_1, P_1, \Phi_1)$ and $\mathcal M_2=(S_2, T_2, P_2, \Phi_2)$ be two PMTS. A binary relation $R^n_{\mu,\nu} \subseteq S_1 \times S_2 \; (n \in \mathbb{N}_0)$ is a \emph{bounded modal refinement relation} under two fixed valuations $\mu \subseteq P_1$ and $\nu \subseteq P_2$ if either $n = O$, then $R^0_{\mu,\nu} = S_1 \times S_2$, or if for every $(s,t) \in R^{n+1}_{\mu,\nu}$ holds
\begin{align*}
\forall M \in \Tran_{\mu}\left(s\right):  \;& \exists N \in \Tran_{\nu}\left(t\right): \\
             & \forall \left(a,s'\right) \in M: \exists \left(a,t'\right) \in N: \left(s', t'\right) \in R^n_{\mu,\nu} \\
\wedge \quad & \forall \left(a,t'\right) \in N: \exists \left(a,s'\right) \in M: \left(s', t'\right) \in R^n_{\mu,\nu} 
\end{align*}
We say that $s$ \emph{$n$-bounded modally refines} $t$, denoted by $s \leq_m^n t$, if for all $\mu \subseteq P_1$ there exists a modal refinement relation $R^n_{\mu,\nu}$ with some $\nu \subseteq P_2$ such that $(s,t) \in R^n_{\mu,\nu}$.
\end{definition}

\begin{lemma}
\label{lemma:mr:bounded}
On finite {PMTS} modal refinement and \emph{bounded} modal refinement coincide, meaning $s \leq_m t$ if and only if $s \leq_m^n t$ for all $n \in \mathbb{N}_0$.
\end{lemma}

\begin{proof}
'If' part. Let's assume $s \leq_m^n t$ for all $n \in \mathbb{N}_0$. Then there exists for every $\mu \subseteq P_1$ a nonincreasing series of sets $R^{0}_{\mu,\nu}, R^{1}_{\mu,\nu}, R^{2}_{\mu,\nu} \dots$ with the \emph{bounded} modal refinement definition applied each time and in all these sets is $(s,t)$ contained. Every iteration of the \emph{bounded} modal refinement definition will either remove at least one element or remove nothing and stabilize. As the underlying PMTS is finite, this series is stable after at most $|S_1 \times S_2|$ iterations and $R_{\mu,\nu} = R^{|S_1 \times S_2|}_{\mu,\nu} \ni (s,t)$ is a sufficient modal refinement relation for $\mu$ and $\nu$. As this is applicable for every $\mu \subseteq P_1$, $s \leq_m t$ holds.

'Only-If' part. Let's assume $s \leq_m t$. Then there exists for every $\mu \subseteq P_1$ a modal refinement relation with $(s,t) \in R_{\mu,\nu}$. Let now $R^{0}_{\mu,\nu}, R^{1}_{\mu,\nu}, R^{2}_{\mu,\nu} \dots$ be a nonincreasing series of sets with each time the \emph{bounded} modal refinement definition applied. Clearly for all $i \in \mathbb{N}_0: (s,t) \in R_{\mu,\nu} \subseteq R^{i}_{\mu,\nu}$. As this this can be done for every $\mu \subseteq P_1$, we have $s \leq_m^n t$ for all $n \in \mathbb{N}_0$.
\qed\end{proof}

\subsection{Proof of Lemma~\ref{lem:tr}}

First, we state a trivial technical claim.

\begin{claim}
$Avoid$ is downward closed, i.e.
\[
(s, \mathcal{T}) \in Avoid \implies \forall \mathcal{T}' \subseteq \mathcal{T}: (s, \mathcal{T}') \in Avoid
\]
\end{claim}

\begin{proof}[of Lemma~\ref{lem:tr}]
'If' part (soundness of the construction). As $Avoid$ is defined as smallest set, let $Avoid_0, \, Avoid_1, \, Avoid_2 \dots$ denote the non-decreasing sequence of sets leading to $Avoid$ by applying the definition each time. We initialize $Avoid_0$ with $(s, \emptyset)$ for all $s \in S$. We prove by induction on $n$ that, whenever $(s, \mathcal{T}) \in Avoid_n$, there exists an implementation $I$ such that $I \mr s$ and $\forall t \in \mathcal{T}: I \not \mr t$.

The base case $n = 0$ is trivial, as $\mathcal{T} = \emptyset$ the underlying {BMTS} is globally consistent and by corollary \ref{tr:cor:consistent} there is an implementation $I$ for $s$. For the induction step assume $(s, \mathcal{T}) \in Avoid_{n+1}$. 

As $(s, \mathcal{T}) \in Avoid_{n+1}$ there exists sets $M \in \Tran_\emptyset(s)$ and $later_{a,s',f}$ such that the conditions of the definition hold. By the second part of the condition and the induction hypothesis, for all $(a, s') \in M$ and $f$ there exists an implementation $I_{a,s',f}$ with $I_{a,s',f} \mr s'$ and $I_{a,s',f} \not \mr t'$ for all $t' \in later_{a,s',f}$. We now construct a new implementation $I$, such that $I \mr s$ and $I \not \mr t$ for all $t \in \mathcal{T}$. We simply take the disjoint union of the previously mentioned $I_{a,s',f}$, add a new state $I$ with new transitions $(I, a, I_{a,s',f})$ for every $(a, s') \in M$ and $f$.

We now show that indeed $I \mr s$ and $I \not \mr t$ for all $t \in \mathcal{T}$ holds. The first claim trivially holds by construction. For the second claim, let us consider some arbitrary $t \in \mathcal{T}$. Then for each $N \in \Tran(t)$ there exists a particular action $a \in \Sigma$, for which one of the disjunctions holds. 

Whenever the first is true then either $M(a)$ is empty, which is a violation of the modal refinement condition, or there exists $t' \in N(a)$, which is contained in $later_{a,s',f}$ for each $s' \in M(a)$. Since $I_{a,s',f} \not \mr t'$ the modal refinement condition is violated.

Whenever the second is true then again either $N(a)$ is empty, which is a direct violation of the modal refinement definition, as $t$ cannot match the move of $I$, or there exists a $s' \in M(a)$, such that $\emptyset\neq N(a) \subseteq later_{a,s',f}$. Since $I_{a,s',f}\not \mr t'$ for all $t' \in later_{a,s',f}$, $(I_{a,s',f}, t')$ is never contained in a modal refinement relation. As $I_{a,s',f}$ is by construction an $a$-successor of $I$, the modal refinement condition is violated for $(I, t)$. Therefore the second claim holds.

'Only-If' part (completeness of the construction). We prove by induction on $n$, that whenever there exists an implementation $I$ with $I \mr s$ and $I \not \mr^n t$ for all $t \in \mathcal{T}$ then $(s, \mathcal{T}) \in Avoid$. After that, lemma \ref{lemma:mr:bounded} is applied. The base case $n = 0$ is trivial, as all pairs of processes are refining each other, hence $\mathcal{T} = \emptyset$ and by definition $(s,\mathcal{T}) \in Avoid$.

For the induction step assume the existence of an implementation $I$, such that  $I \mr s$ and $I \not \mr^{n+1} t$ for all $t \in \mathcal{T}$. As $I$ is an implementation, $\Tran(I)$ is a singleton and $N_I \in \Tran(I)$ is unique. Furthermore as $I \mr s$ holds, there exists by definition $N_s \in \Tran(s)$, such that for all $a \in \Sigma$
\begin{enumerate}
\item $\forall s_a \in N_s(a): \exists I_a \in N_I(a): I_a \mr s_a$
\item $\forall I_a \in N_I(a): \exists s_a \in N_s(a): I_a \mr s_a$
\end{enumerate}

To show that $(I,\mathcal T)\in Avoid_{n+1}$ we set $M:=N_s$. For each $f\in\bigcup_{t\in\mathcal T}\Tran(t)$ and each $(a,s')\in N_s$, we define $later_{a,s',f}$ such that the conditions are satisfied. We set 
\begin{align*}
later_{a,s',f}:=\{t'\mid \exists t\in\mathcal T:& \exists N_t\in \Tran(t):t'\in N_t(a)\wedge \tag{*}\\
&\forall I'\in N_I(a):I'\not\mr^{n} t' \\ \vee & f=N_t \wedge \exists I'\in N_I(a): I'\mr s' \wedge \forall t''\in N_t(a): I'\not\mr^{n} t''\} 
\end{align*}

Let $t \in  \mathcal{T}$ and $N\in \Tran(t)$ be arbitrary but fixed. As $I \not \mr^{n+1} t$, for some $a \in \Sigma$, there is a  violation of the modal refinement definition, such that one of the cases hold:
\begin{enumerate}
\item $N_I(a) = \emptyset \wedge N_t(a) \neq \emptyset$
\item $N_I(a) \neq \emptyset \wedge N_t(a) = \emptyset$
\item $\exists t_a \in N_t(a): \forall I_a \in N_I(a): I_a \not \mr^n t_a$
\item $\exists I_a \in N_I(a): \forall t_a \in N_t(a): I_a \not \mr^n t_a$
\end{enumerate}

If the third holds, then due to the first disjunct of (*) we can satisfy the first disjunct of Definition by giving the same $t_a$.
If the fourth holds, then due to the second disjunct of (*) we can satisfy the second disjunct of Definition for any $s'$ with $I'\mr s'$ (there is one due to 2.).

Finally, to prove that $(s_a, later_{a,s_a,f})$ has a $n$-step distinguishing implementation it is sufficient to take $I'$ of the second disjunct.

\qed\end{proof}

\section{Thorough vs.\ Modal Refinement}\label{app:mrtr}

\subsection{Proof of Proposition~\ref{prop:rightdet}}

\begin{proof}
We fix a valuation $\nu$ of parameters and define a~relation $R$ that satisfies the condition of Definition~\ref{def:mr-bmts}.
The relation $R$ is taken as the smallest relation such that
$(s_0^\nu, t_0) \in R$ and whenever
$(s, t) \in R$, $(s,a,s')\in T$ and $(t,a,t')\in T$ then also $(s',t') \in R$.
Before we prove that $R$ satisfies the conditions,
we make the claim that $(s,t) \in R$ implies $s \tr t$.
Clearly, this holds for $(s_0,t_0)$. Suppose now that $s \tr t$, 
$(s,a,s'),(t,a,t')\in T$ and $i'$ is an~arbitrary implementation of~$s'$.
Then there exists an implementation $i \in \semantics{s}$ such that
$i \must{a} i'$. But as $s \tr t$, $i$ is also an~implementation of $t$.
Therefore, as $t$ is deterministic, $i'$ is an~implementation of $t'$,
thus $s' \tr t'$. We can now check that $R$ satisfies the condition of Definition~\ref{def:mr-bmts}.
Let $(s,t) \in R$ and $M\in \Tran{s}$. Define $A:=\{a\mid \exists s':(a,s')\in M\}$. There is an implementation $i$ with exactly transitions under $A$. Moreover, according to the assumption it also an implementation of $t$. Hence $N:=\{(a,t')\mid (t,a,t')\in T \wedge a\in A\}$ is an element of $\Tran(t)$. The two conjuncts then clearly hold by construction of $R$.
\qed\end{proof}

\subsection{Proof of Proposition~\ref{prop:hull-sound}}

\begin{proof}
As the transition system of $\deth(\xmts)$ is created by the powerset construction, it is clearly deterministic. We prove that $s_0\mr \deth(s_0)$. Since both systems have the same parameter set, for any valuation of parameters of $\xmts$ we can choose the same valuation for $\deth(\xmts)$. Further, we define relation $R$ such that $(s, S) \in R$ iff $s \in S$ and show that the condition of
Definition~\ref{def:mr-bmts} is satisfied. Let $(s,S)\in R$. For $M\in\Tran(s)$, we set $N:=M[(a,S_a)/(a,s')]$. Since $\Tran(S)=\bigcup_{s\in S}\Tran(s)[(a,S_a)/(a,s')]$, we have $N\in\Tran(S)$. We check the two conjuncts. Whenever there is $(a,s')\in M$ then $(a,S_a)\in N$ and $s'\in S_a$ hence $(s',S_a)\in R$. Whenever there is $(a,S_a)\in N$ we have the respective $(a,s')\in M$ with by construction of $N$. Further, $s'\in S_a$ hence $(s',S_a)\in R$.
\qed\end{proof}

\subsection{Proof of Proposition~\ref{prop:hull-complete}}

\begin{proof}
Assume $t_0$ deterministic state of a PMTS $\mathcal N$ with $s_0\mr t_0$. 
Therefore, there for every valuation $\mu$ there is a valuation $\nu$ and the greatest relation $R_{\mu,\nu}$ containing $(s_0^\mu,t_0^\nu)$ and satisfying the condition of Definition~\ref{def:mr-bmts}. We show that $\deth(s_0) \mr t_0$ by choosing for every $\mu$ the same $\nu$ and 
constructing a~new relation $Q_{\mu,\nu}$ between states of $\deth(\xmts)^\mu$ and $\mathcal N^\nu$ that also satisfies this
condition as follows:
$$(S, t) \in Q_{\mu,\nu} \quad\text{if and only if}\quad \emptyset \not= S \subseteq \{ s \mid (s, t) \in R_{\mu,\nu} \}$$
We now check the condition. Since $(s_0,t_0)\in R_{\mu,\nu}$, we have $(\deth(s_0)^\mu,t_0^\nu)=(\{s_0\}^\mu,t_0)=(\{s_0^\mu,t_0\})\in
Q_{\mu,\nu}$. Let now $(S,t)\in Q_{\mu,\nu}$ and $M\in\Tran(S)$. Hence there is $s\in S$ with $M'\in\Tran(s)$ with $M=M'[(a,S_a)/(a,s') \text{ for every $a,s'$}]$. Since $(s,t)\in R_{\mu.\nu}$, there is $N\in\Tran(t)$ matching $M'$. We show it also matches $M$. Let $(a,S_a)\in M$. There is unique (due to determinism) $(a,t')\in N$ and further $(S_a,t')\in Q_{\mu,\nu}$ as each $s_a\in S_a$ modally refines the only $a$-successor of $t$, thus $(s_a,t')\in R_{\mu,\nu}$. Similarly, let $(a,t')\in\Tran(t)$. Then there is unique $(a,S')\in\Tran(S)$, namely $(a,S_a)$. For the same reasons as above $(s_a,t')\in R_{\mu,\nu}$ for every $s_a\in S_a$.

The minimality for BMTS holds w.r.t.~both thorough and modal refinements as they coincide when the refined system is a deterministic BMTS.
%
\qed\end{proof}

\subsection{Proof of Proposition~\ref{prop:phull-sound}}

\begin{proof}
First, observe that for any parameter valuation $\nu$, the identity relation satisfies the condition of Definition~\ref{def:mr-bmts} for $s_0^\nu$ and $\parh(s_0)$. Indeed, for any $M\in\Tran(s)$ we also have $M\in\Tran(\parh(s))$. Similarly, $\{((s,\nu),s)\mid s\in S,\nu\subseteq P\}\cup\{(s_0^B,\parh(s_0))\}$ satisfies the condition for $s_0^B$ and $\parh(s_0)$. 
\qed\end{proof}

\end{document}